\documentclass[sn-mathphys-num]{sn-jnl}

\usepackage{graphicx}%
\usepackage{multirow}%
\usepackage{amsmath,amssymb,amsfonts}%
\usepackage{amsthm}%
\usepackage{mathrsfs}%
\usepackage[title]{appendix}%
\usepackage{xcolor}%
\usepackage{textcomp}%
\usepackage{manyfoot}%
\usepackage{booktabs}%
\usepackage{algorithm}%
\usepackage{algorithmicx}%
\usepackage{algpseudocode}%
\usepackage{listings}%
\usepackage{verbatim}

\theoremstyle{thmstyleone}%
\newtheorem{theorem}{Theorem}
\theoremstyle{thmstyletwo}%

\theoremstyle{thmstylethree}%

\begin{document}

\title{Shifted Eigenvector Models for Centrality and Occupancy in Urban Networks}

\author[1]{\fnm{María Magdalena} \sur{Martínez-Rico}}

\author*[2]{\fnm{Luis Felipe} \sur{Prieto-Martínez}}\email{luisfelipe.prieto@upm.es}

\equalcont{These authors contributed equally to this work.}

\affil[1]{\orgaddress{\city{Aranjuez}, \postcode{28300}, \state{Madrid}, \country{Spain}}}

\affil[2]{\orgdiv{Departmento de Matemática Aplicada}, \orgname{Universidad Politécnica de Madrid}, \orgaddress{\city{Madrid}, \postcode{28040}, \state{Madrid}, \country{Spain}}, \url{https://orcid.org/0000-0002-7825-4743}}

\abstract{This article investigates a family of centrality models for urban networks that incorporate both topological and non-topological factors. Since centrality is inherently recursive, these models can be formulated as fixed-point equations, which we refer to as \emph{shifted eigenproblems}. Assuming a correlation between node centrality and occupancy, we discuss how experimental data can be used to estimate model parameters via least-squares methods. Furthermore, such data would allow us to infer the \emph{intrinsic attraction} of each node, as well as the \emph{occupancy induced by must-visit points of interest}, a task that is conceptually challenging. Once the model parameters are fitted and validated, our framework can be used to assess the impact of urban interventions, such as introducing a \emph{must-visit point of interest} at a specific node or enhancing its \emph{intrinsic attraction}. The resulting sensitivity analysis is therefore highly relevant for urban planning decisions. We also provide explicit formulas to facilitate this analysis.
\medskip

\noindent \textbf{MSC codes:} 0A67, 05C50.
}

\keywords{Graph centrality, Urban networks, Occupancy models, Spectral centrality, Urban planning.}

\maketitle

\section{Introduction}

\subsection*{Centrality and urban networks}

Cities are complex systems where spatial configuration, human movement, and social interaction intertwine. Understanding why certain streets, intersections, or urban spaces become \emph{more central} than others is crucial not only for mobility and accessibility but also for urban design and the evolution of city form. In this paper, we examine these questions through the lens of network theory and mathematical centrality, bridging ideas from architecture, urbanism, and applied mathematics.

We model an urban network as a graph $(V, E)$ with vertex set $V = \{V_1, \ldots, V_n\}$, where nodes represent \emph{streets} or other \emph{spatial entities}. For technical details on this kind of graph construction, including \emph{segment maps} within \emph{Space Syntax}, we refer the reader to \cite{Crucitti2006, HH, JC, Porta2006}. Throughout, we assume $(V,E)$ is connected. For each $i=1,\ldots,n$, let $x_i$ denote the \emph{centrality} of node $V_i$, and $\mathbf{x} = [x_1,\ldots,x_n]^T$ the \textbf{vector of centralities}.

Note that an urban network can also represent enclosed spaces such as a museum. We use such closed examples in this paper because they are, in some sense, more tractable and provide a controlled setting for testing our models.

In this context, \emph{centrality} quantifies the relative importance of nodes within a network, influencing both movement patterns and the spatial experience of the city. Classical indices—such as degree, closeness, betweenness, and eigenvector centrality—capture different aspects of connectivity, accessibility, and bridging roles \cite{bonacich1972factoring, bonacich1987power, Katz1953}. Centrality analysis has proven invaluable in urban research, illuminating mobility, land-use distribution, and the historical evolution of urban form \cite{Crucitti2006, Porta2006, JC}.

A particularly intriguing aspect of centrality is its recursive nature: \emph{a node is central not only because of its immediate connections, but also because it connects to other central nodes}. This \emph{feedback mechanism} naturally leads to eigenvector-based formulations \cite{agryzkov2019centrality,bonacich1987power,Katz1953,Nourian2016} grounded in nonnegative matrix theory \cite{BP,Chung1997,horn2012matrix,Seneta2006,varga2009matrix}.

For the nodes of a given urban network $(V, E)$, two fundamental metrics can be defined: \emph{occupancy} and \emph{total footfall}. \textbf{Occupancy} refers to the instantaneous number of individuals present in a given node at a specific moment, reflecting how populated that space is at that time. In contrast, \textbf{total footfall} captures the cumulative number of people who enter or traverse that node over a defined time interval, thus representing the overall volume of movement through the space. Although related, these measures describe distinct aspects of spatial dynamics. For instance, a transit hall or passageway may exhibit very high total footfall due to continuous pedestrian flows, while maintaining relatively low occupancy levels because individuals typically spend only a short time there.

In an idealized scenario, centrality scores would mirror both the occupancy and total footfall of spaces, with the most central streets or intersections corresponding to the most heavily trafficked locations. Of course, such a perfect correspondence is unlikely in real-world settings, where additional factors such as behavioral patterns, temporal variations, and local context intervene. Nevertheless, the degree of correlation between centrality and occupancy provides a natural criterion for estimating empirical centrality values for each node. In this paper, however, we focus solely on occupancy, as it is easier to measure in small, toy examples.

\subsection*{Outline of the paper}

In the following section, we present a toy example to illustrate the concepts and methods appearing in this work. After this, the paper is structured around three main parts:

\begin{itemize}

\item[(1)] \textbf{Model formulation (Section \ref{section.eigen}).}  We present a family of models for centrality in urban networks. These models are widely used in several contexts in applied mathematics and we think that they are also applicable to urban networks. We assume that the centrality of a given node in some urban network depends on both topological and non-topological factors. As we will discuss later, the non-topological factors concern the \emph{intrinsic attraction} of the corresponding space (presence of points of interest, aesthetic qualities and other contextual factors)  and/or the presence of \emph{must-visit points of interest} (such as workplaces). The \emph{feedback mechanism} mentioned before, naturally leads to some fixed point models (\emph{eigenvector models} and \emph{shifted eigenvector models}).

\item[(2)] \textbf{Parameter estimation and validation (Section \ref{section.flow}).}  If we assume that the measure of centrality of a given node of some urban network should mirror the occupancy of the corresponding space, then we can easily estimate the parameters of any of the models discussed in Section \ref{section.eigen} from real data using  least-squares methods. Moreover, experimental data allow us to estimate the \emph{intrinsic attraction} of a given space and the occupancy due to \emph{must-visit points of interest}, which is something conceptually difficult. We present some strategies for fitting the parameters and validating the models.

\item[(3)] \textbf{Sensitivity analysis (Section \ref{section.sensitivity}).} Once we have selected a model, we have fixed its parameters and we have checked that this fitting is adequate, we can use our approach to study the impact of some urban intervention at a given node of the type \emph{adding a must-visit point of interest} or \emph{improving its intrinsic attraction}). The sensitivity analysis of our models is, then, very relevant in urban planning.  We present some formulas to perform this sensitivity analysis (Theorems \ref{theo.sensitivity}, \ref{theo.2}).

\end{itemize}

\section{A toy example} \label{section.toy}

Consider a university building whose floor plan has the shape of a letter $C$, as shown in Figure \ref{fig.ETSAM}. To simplify, we model it as three spaces represented by nodes $V_1$, $V_2$, and $V_3$, where $V_1$ is connected to $V_2$ and $V_2$ to $V_3$, forming a simple linear graph.

\begin{figure}[h]
\centering
\includegraphics[width=0.4\textwidth]{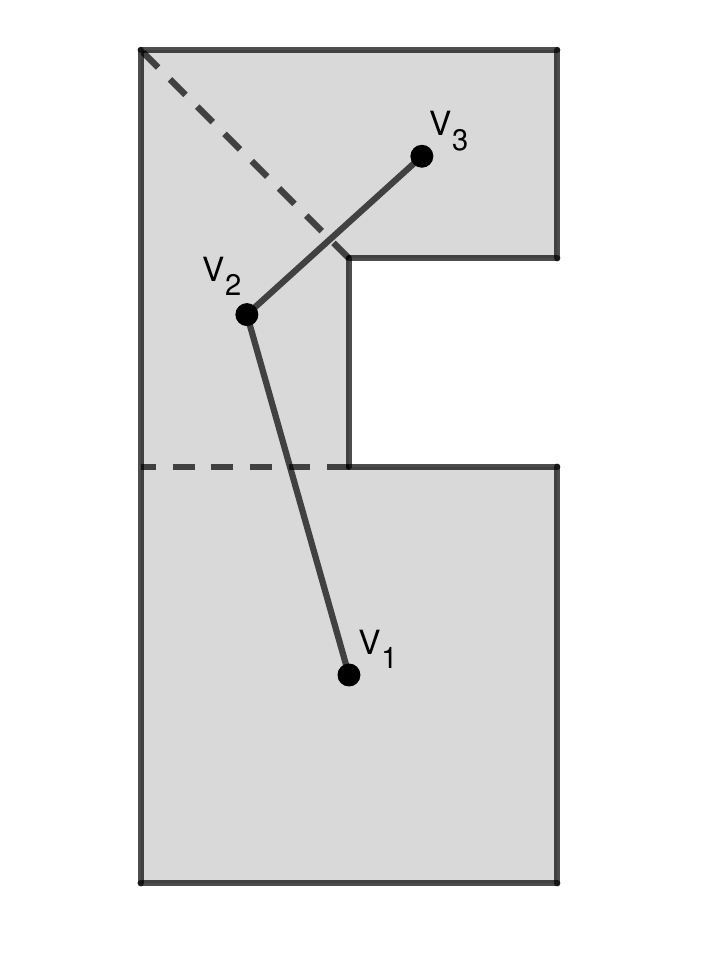}
\caption{Floor plan of the building described in the example.}\label{fig.ETSAM}
\end{figure}

For this graph, the \textbf{adjacency matrix} (obtained writing 1 in the position $(i,j)$ if and only if $V_i,V_j$ are connected and 0 otherwise) is
\begin{equation} \label{eq.adj}
\mathbf A=\begin{bmatrix}0 & 1 & 0 \\ 1 & 0 & 1 \\ 0 & 1 & 0 \end{bmatrix}
\end{equation}

In the following sections, we will describe different features of these three nodes, that are related to  their \emph{intrinsic attraction} and possible \emph{must-visit points of interest}, when necessary.


\section{Shifted eigenvector models for centrality} \label{section.eigen}

In the following, for any square matrix $\mathbf M$, let us denote by $\rho(\mathbf M)$ to its \textbf{spectral radius}, that is, the maximum absolut value of the eigenvalues of $\mathbf M$.

\subsection*{Spectral centrality and Perron-Frobenius Theorem}

Perhaps the most illustrative example of a centrality measure in which the \emph{feedback mechanism} described in the introduction leads to an eigenproblem is \emph{spectral centrality}, a classical concept in network analysis \cite{bonacich1972factoring, Newman2010}. Intuitively, it captures how a street or intersection gains importance not only from its immediate connections, but also from the importance of the locations to which it is linked. Formally:

\medskip

\noindent\textbf{Spectral centrality.} \emph{Let $\mathbf A$ denote the \textbf{adjacency matrix} of an urban network $(V,E)$. The \textbf{vector of spectral centralities} is defined as the (unique up to scaling) positive eigenvector of $\mathbf A$, that is, the (unique up to scaling) positive vector satisfying}
\begin{equation}\label{eq.easy}
\mathbf A \mathbf x = \lambda \mathbf x,\qquad \text{for some positive real }\lambda.
\end{equation}

\medskip

Let us remark that this measure of centrality depends only on topological aspects (connectivity) and does not take into acount any other feature of the spaces corresponding to the nodes. Also, it is not an absolut, but a relative measure, only valid to compare the nodes.

In the toy example described in Section \ref{section.toy}, the unique (up to scaling) positive eigenvector is $\mathbf x=[1,\sqrt{2},1]^T$. It corresponds to the eigenvalue $\lambda=\sqrt{2}$.

Since the urban network is connected, then the adjacency matrix $\mathbf A$ is irreducible. So, Eigenproblem \eqref{eq.easy} is guaranteed to have a unique (up to scaling) positive solution $\mathbf x$ by the classical result known as \emph{Perron-Frobenius Theorem} \cite{Seneta2006, ersel2011linear}. Since this result is recalled several times in this paper, let us include one of its possible statements (it is a particularization of  Theorem 10.12 in \cite{ersel2011linear}, so we refer the reader to that article and the references therein). 

\begin{theorem}[Perron--Frobenius Theorem] \label{theo.PF}
Let $\mathbf M$ be a nonnegative, irreducible matrix.

\begin{itemize}

\item There exists a real positive eigenvalue $\lambda$ for $\mathbf M$, called \textbf{Perron-Frobenius eigenvalue}.

  \item The Perron-Frobenius eigenvalue equals the $\rho(\mathbf M)$ and it is \textbf{algebraically simple} (its algebraic multiplicity is $1$).

  \item The eigenvector $\mathbf x$ corresponding to the Perron-Frobenius eigenvalue is unique up to scaling and it can be chosen strictly positive. We call any positive eigenvector corresponding to the Perron-Frobenius eigenvalue a \textbf{Perron-Frobenius eigenvector}.

  \item Perron-Frobenius eigenvectors are the unique positive eigenvectors of $\mathbf M$ (that is, if $\mathbf y$ is a positive vector and satisfies $\mathbf M\mathbf y=\mu \mathbf y$, then $\mu=\lambda$ and $\mathbf y$ is a Perron-Frobenius eigenvector).

\end{itemize}
\end{theorem}

Not only the existence and uniqueness (up to scaling) of solutions is guaranteed, but several exact and numerical methods are known to find such an eigenvector. For a detailed discussion, see \cite{strang2009introduction,Golub2013}.

\subsection*{Variations of spectral centrality: eigenvector models for centrality}

Many interesting variations of this model have been proposed to address its limitations \cite{bonacich1972factoring, agryzkov2019centrality,Newman2010, erlander1990gravity}. Some notable examples include the considerations in the following list:

\begin{itemize}
\item The classical model considers only the immediate neighbors of each node. It is possible to account for nodes at greater distances, assigning less importance to those further away. This can be achieved by replacing $\mathbf A$ in Equation \eqref{eq.easy} with any of the (hollow, nonnegative, and symmetric) matrices
\[
\mathbf E = 
\begin{bmatrix} 
0 & \frac{1}{d_{12}} & \ldots & \frac{1}{d_{1n}} \\ 
\frac{1}{d_{21}} & 0 & \ddots & \vdots  \\ 
\vdots & \ddots & \ddots & \frac{1}{d_{n-1,n}} \\ 
\frac{1}{d_{n1}} & \ldots  & \frac{1}{d_{n,n-1}} & 0
\end{bmatrix}, \qquad
\widetilde{\mathbf E} = 
\begin{bmatrix} 
0 & \frac{1}{d_{12}^2} & \ldots & \frac{1}{d_{1n}^2} \\ 
\frac{1}{d_{21}^2} & 0 & \ddots & \vdots  \\ 
\vdots & \ddots & \ddots & \frac{1}{d_{n-1,n}^2} \\ 
\frac{1}{d_{n1}^2} & \ldots  & \frac{1}{d_{n,n-1}^2} & 0
\end{bmatrix}.
\]
Here, the positive values $d_{ij}$ may represent (i) metric distances between street centroids, (ii) travel times, or (iii) shortest-path lengths on the graph, depending on the modeling context. These formulations relate to well-known concepts such as \emph{weighted harmonic centrality} or \emph{gravity models} \cite{Newman2010, erlander1990gravity}.

\item As noted by several authors, it is possible to \emph{weight} the nodes to account for their relative importance or \emph{intrinsic attraction} \cite{agryzkov2019centrality}.

Continuing with the toy example from Section \ref{section.toy}, suppose that in the space represented by node $V_1$, students and faculty have access to bathrooms, vending machines, a water fountain for filling bottles, and some desks for studying or working; while in the space represented by node $V_3$, a second pair of bathrooms is available. The different characteristics of the spaces represented by $V_1$, $V_2$, and $V_3$ lead to different weights, $w_1$, $w_2$, and $w_3$, reflecting these variations. This extension of the toy example will be revisited in Section \ref{section.flow}.

To construct a model that accounts for this type of phenomenon, we can replace $\mathbf A$ in Equation \eqref{eq.easy} with $\mathbf A \mathbf W$ (or with $\mathbf E \mathbf W$ or $\widetilde{\mathbf E} \mathbf W$, if combining with the previous approach), where $\mathbf W$ is a diagonal matrix with entries $w_1, \ldots, w_n$ corresponding to the weight of each node $V_i$.
\end{itemize}

In any of the variations proposed above, we ultimately need to solve an eigenproblem of the form:

\medskip
\noindent \textbf{Eigenvector models for centrality.} \emph{Let $\mathbf M$ be a nonnegative irreducible matrix suitably associated to an urban network $(V,E)$. We define the \textbf{eigenvector of centralities associated to $\mathbf M$} as the (unique up to scaling) positive eigenvector of $\mathbf M$, that is, the (unique up to scaling) positive vector satisfying}
\begin{equation}\label{eq.general}
\mathbf M \mathbf x = \lambda \mathbf x\qquad \text{for some positive real }\lambda.
\end{equation}

\noindent \emph{ In most of the cases, we \emph{normalize}, dividing the matrix by the (unique) dominant eigenvalue $\lambda$ (see below), and we consider, instead, models of the type}
\begin{equation}\label{eq.generalnorm}
\mathbf M \mathbf x = \mathbf x,
\end{equation}

\noindent and impose that $\rho(\mathbf M)=1$.

\medskip

The matrix $\mathbf M$ is sometimes referred in the literature concerning these kind of models as the \textbf{matrix of relationships} \cite{bonacich1972factoring}.

Let us remark that, if the urban network is connected and  for $\mathbf W$ being a diagonal matrix with positive weights in its main diagonal, the matrices $\mathbf A,\mathbf E,\widetilde{\mathbf E},(\mathbf A\mathbf W),(\mathbf E\mathbf W),(\widetilde{\mathbf E}\mathbf W)$ that we have considered above are irreducible and so Perron-Frobenius Theorem applies.

\subsection*{Mandatory visit to the nodes: shifted eigenvector models for centrality}

In the context of our toy example from Section \ref{section.toy}, let us suppose that some classes take place in the spaces represented by the nodes $V_2$ and $V_3$. In this case, the previous \emph{eigenvector models for centrality} are not adequate: we need to incorporate into the model (for measuring centrality or occupancy) the fact that students and faculty \emph{must visit} the classrooms.

\medskip

\noindent \textbf{Shifted eigenvector models.} \emph{Let $\mathbf M$ be an irreducible, nonnegative matrix with $\rho(\mathbf M)<1$ and $\mathbf f$ a nonnegative column vector. We define the \textbf{centrality eigenvector corresponding to $\mathbf M, \mathbf f$} to be the (unique) vector $\mathbf x$ satisfying}
\begin{equation}\label{eq.eigshift}
\mathbf x = \mathbf M \mathbf x + \mathbf f.
\end{equation}

\medskip

We refer to Problems as the one in Equation \eqref{eq.eigshift} as \textbf{shifted eigenproblems}, interpreting that the vector $\mathbf f$ is a \emph{shifting}.

In this case, centrality is an absolute measure. Assuming that $\mathbf x$ represents not only the centralities but the occupancies of the corresponding spaces, then we can suppose that $\mathbf f$ measures the occupancy of the node, due to the presence of must-visit points of interest, that is, the \textbf{forced occupancy}. Then, only the \emph{free-to-move people} relocate according to the connection between spaces and their \emph{intrinsic attraction}. So, in the previous model, as it has been stated, $\mathbf M$ needs to be normalized so that $\mathbf 1\mathbf M+\mathbf 1\mathbf f$ equals the total occupancy of the urban network, where $\mathbf 1=[1,\ldots, 1]$. This idea is reviewed in Section \ref{section.sensitivity}.

These models  resemble some already appearing in other contexts \cite{BP,ersel2011linear, Newman2010}. The mathematical treatment of Problem \eqref{eq.eigshift} is well known in the economic literature related to \emph{Leontief's Input-Output Model}. The mathematical basis for them is the following result:

\begin{theorem}[Adapted from Lemma 2 in \cite{BP}] \label{theo.leontief} 
Let $\mathbf M$ be an irreducible, square, nonnegative matrix and $\mathbf f$ a nonnegative vector.
\begin{itemize}
\item[(a)] If the system $(\mathbf I-\mathbf M)\mathbf x=\mathbf f$ has a nonnegative and nontrivial solution $\mathbf x$, then $\mathbf x$ is positive $\rho(\mathbf M)\leq 1$.
\item[(b)] If $\rho(\mathbf M)=1$, then $(\mathbf I-\mathbf M)\mathbf x=\mathbf f$ has a nonnegative solution $\mathbf x$ if and only if $\mathbf f=\mathbf 0$ (the case of Theorem \ref{theo.PF}). 
\item[(c)] If $\rho(\mathbf M)<1$, then $(\mathbf I-\mathbf M)$ is invertible and its inverse is nonnegative. So, $(\mathbf I-\mathbf M)\mathbf x=\mathbf f$ has a unique nonnegative solution $\mathbf x$ for every $\mathbf f$. This solution is nontrivial if and only if $\mathbf f\neq \mathbf 0$, and in this case, it is positive.
\end{itemize}
\end{theorem}

\section{Estimation of the parameters and validation of the models}\label{section.flow}

Let $\mathbf B$ be any of the matrices $\mathbf A,\mathbf E,\widetilde{\mathbf E}$ described in the previous section.  For those models corresponding to a matrix of relationships of the type $\mathbf M=\mathbf B\mathbf W$, where $\mathbf W$ is a diagonal matrix encoding the weights describing the \emph{intrinsic attraction} of the nodes, we have a fundamental and conceptual problem: how do we estimate the entries in $\mathbf W$?

As explained in the introduction, the ideal correlation between centrality and occupancy provides a natural criterion to answer empirically the previous question from real data.


For the rest of this section and for a fixed urban network $(V,E)$, let $\mathbf x_1,\mathbf x_2,\ldots,\mathbf x_N$ be some real data.  For each $i=1,\ldots, n$, $\mathbf x_i$ encodes the real occupancy of the nodes at a given moment. Let us also denote by $\mathbf X_i$ to the diagonal matrix corresponding to $\mathbf x_i$. Collecting this data may be challenging for big urban networks and some authors propose the use of mobile device data \cite{niu2023technical}.

Let us consider shifted eigenvector models, corresponding to those matrices of relationships discussed at the beginning of this section,  of the type
\begin{equation} \label{eq.2e} 
\mathbf x=\underbrace{\mathbf B \mathbf W}_{\mathbf M} \mathbf x +\mathbf f.\end{equation}

\noindent It is possible to estimate the node weights $\mathbf{W}$ from the observed data $\mathbf{x}_1, \ldots, \mathbf{x}_N$. The following subsections describe the procedure, distinguishing between the case where the forced occupancy vector $\mathbf{f}$ is known and the case where it is not. In any case, these are \emph{ least-squares fitting problems}. For general methods to solve this type problems, see \cite{strang2009introduction,strang2016la4e,searle1997linear,bjorck1996numerical}. Additionally, the models can be validated by assessing whether the fit is reasonable. Details are omitted here and we refer the reader to previous works for further information.

The methods in this section are based on the fact that the emph{intrinsic atraction} of the nodes is constant.

\subsection*{Estimating the weights if the vectors of forced occupancy are known}

In this subsection we want to estimate the weights in Equation \eqref{eq.2e}. We suppose that, for each vector of real data $\mathbf x_1,\ldots, \mathbf x_N$, we have some vector $\mathbf d_i $ encoding the occupancy of the nodes, due to forced occupancy. Note the vectors $\mathbf d_1,\ldots, \mathbf d_N$ may not be equal, that is, the forced occupancy is not suppose to be constant.

We consider
\begin{equation}\label{sys.over}
\begin{bmatrix}
\mathbf B \mathbf X_1 \\ \hline
\vdots \\ \hline
\mathbf B \mathbf X_N
\end{bmatrix} 
\mathbf w =
\begin{bmatrix}
\mathbf x_1-\mathbf f_1 \\ \hline
\vdots \\ \hline
\mathbf x_N-\mathbf f_N
\end{bmatrix},
\end{equation}

\noindent and then we fit the parameters using the  least-squares method.

The toy example introduced in Section \ref{section.toy} illustrates this framework. Suppose that we choose $\mathbf B=\mathbf A$ (the adjacency matrix in Equation \eqref{eq.adj}). Let us suppose that at the nodes $V_2,V_3$ some classes are taking place. Then some students and faculty count as \emph{forced occupancy} of the nodes.

 Suppose that, as experimental data, we have obtained the vectors
\[
\mathbf x_1 = \begin{bmatrix} 105 \\ 297 \\ 98 \end{bmatrix}, \qquad
\mathbf x_2 = \begin{bmatrix} 99 \\ 303 \\ 98 \end{bmatrix}, \qquad
\mathbf x_3 = \begin{bmatrix} 97 \\ 289 \\ 113 \end{bmatrix},
\]

\noindent and
\[
\mathbf f_1 = \begin{bmatrix} 0 \\ 297 \\ 95 \end{bmatrix}, \qquad
\mathbf f_2 = \begin{bmatrix} 0 \\ 300 \\ 95 \end{bmatrix}, \qquad
\mathbf f_3 = \begin{bmatrix} 0 \\ 289 \\ 95 \end{bmatrix}.
\]

\noindent This is because at nodes $V_2,V_3$ some classes are taking place and we can count the number of students assisting to these classes. To estimate the parameters $w_1, w_2, w_3$ encoding the \emph{intrinsic attraction} of the nodes $V_1,V_2,V_3$, we can use a \emph{ least-squares fitting} for
\[
\begin{bmatrix}
\mathbf A \mathbf X_1 \\ \hline
\mathbf A \mathbf X_2 \\ \hline
\mathbf A \mathbf X_3
\end{bmatrix} 
\mathbf w =
\begin{bmatrix}
\mathbf x_1-\mathbf f_1 \\ \hline
\mathbf x_2 -\mathbf f_2\\ \hline
\mathbf x_3 -\mathbf f_3
\end{bmatrix},
\]
where, for $i=1,2,3$, $\mathbf X_i$ is the diagonal matrix whose diagonal entries are those in $\mathbf x_i$. That is,
\[
\begin{bmatrix}
0 & 297 & 0 \\ 105 & 0 & 98 \\ 0 & 297 & 0 \\ \hline
0 & 303 & 0 \\ 99 & 0 & 98 \\ 0 & 303 & 0 \\ \hline
0 & 289 & 0 \\ 97 & 0 & 113 \\ 0 & 289 & 0
\end{bmatrix} 
\begin{bmatrix} w_1 \\ w_2 \\ w_3 \end{bmatrix} =
\begin{bmatrix} 105 \\ 0 \\ 3 \\ \hline 99 \\ 3 \\ 3 \\ \hline 97 \\ 0 \\ 8 \end{bmatrix}.
\]

\noindent In this case, we obtain approximately
\[
w_1 = 2, \quad w_2 = 1/3, \quad w_3 = 1.
\]

\subsection*{Estimating the weights if the vector of forced occupancy $\mathbf f$ is not known}

In this case, we want to estimate, not only $\mathbf W$, but $\mathbf f$. To do so, we need to consider that $\mathbf f$ is constant for the real data vectors $\mathbf x_1,\ldots, \mathbf x_N$. This may be the case in some real and larger urban networks, then we need to slightly modify the overdetermined System \ref{sys.over}:
\[
\left[\begin{array}{c | c}
\mathbf B \mathbf X_1  &\mathbf I\\ \hline
\vdots & \vdots \\ \hline
\mathbf B \mathbf X_N & \mathbf I
\end{array} \right]
\begin{bmatrix} \mathbf w\\ \hline \mathbf f\end{bmatrix} =
\begin{bmatrix}
\mathbf x_1\\ \hline
\vdots \\ \hline
\mathbf x_N
\end{bmatrix},
\]

\noindent where $\mathbf I$ is the identity matrix of the adequate size. Again, we can use a \emph{ least-squares fitting} to estimate the parameters and to discuss the model.

\section{Sensitivity analysis and its importance in urban intervention planning} \label{section.sensitivity}

\subsection*{Introduction to sensitivity analysis}
\emph{Sensitivity analysis} studies how variations in a model's input parameters affect its outputs. By systematically perturbing inputs and observing the resulting changes, one can identify which parameters most strongly influence outcomes. This provides insight into model robustness, highlights sources of uncertainty, and guides which aspects require more precise measurement or careful calibration.

A common approach to sensitivity analysis uses partial derivatives of the model outputs with respect to input parameters. The information can be expressed as \emph{elasticities}, obtained by normalizing the derivatives by the ratio of output to input, which allows comparison across parameters with different units or scales. If we want to measure how a variable $y$ varies with respect to some parameter $t$, the \textbf{elasticity} is defined as
\begin{equation}\label{eq.defielas}E_{y,t} = \frac{\partial y}{\partial t} \cdot \frac{t}{y}.\end{equation}

In the context of urban network centrality and occupancy models, sensitivity analysis can reveal how changes in node weights, distance measures, or connectivity affect predicted centrality scores and, consequently, predicted occupancy of the spaces (if we suppose that they are correlated).

Consider planning an urban intervention at a given node, of the type \emph{improving its intrinsic attraction} of \emph{adding a must-visit point of interest}. In the models proposed in Section \ref{section.eigen}, this corresponds to changing a single parameter, such as a node's weight or a forced occupancy coefficient. Elasticities provide a quantitative basis for choosing intervention locations. Locations with high elasticity values are likely to yield the greatest impact on the network, while nodes with low elasticity values would have minimal effect.


Let us suppose that $\mathbf x$ is the vector of centralities corresponding to any of the models described above. To perform the sensitivity analysis of $\mathbf x$ with respect to any of the parameters in the model, \emph{we assume that the sum of the elements in $\mathbf x$ remains constant as this parameter varies.} As a consequence, if $\mathbf x'$ denotes the partial derivative of $\mathbf x$ with respect to this selected parameter, the sum of the entries in $\mathbf x'$ is 0.  In other words,
\begin{equation}\label{eq.ones}\mathbf 1\mathbf x\text{ is constant and }\mathbf 1\mathbf x'=0,\qquad \text{for }\mathbf 1=[1,\ldots, 1].\end{equation}

\subsection*{Sensitivity analysis for eigenvector models with weights}

Let us first consider those (unshifted) eigenvector models for centrality introduced in Section \ref{section.eigen}, having weights. That is, we consider an eigenproblem of the type
\begin{equation} \label{eq.new} \mathbf x=\mathbf M\mathbf x,\qquad \mathbf M=\mathbf B\mathbf W\end{equation}

\noindent  where $\mathbf W$ is a diagonal matrix whose diagonal entries are the (positive) weights $w_1,\ldots, w_n$, and \(\mathbf B\) is a symmetric, irreducible, and nonnegative matrix (this includes the models in Equation \eqref{eq.generalnorm}, for $\mathbf M=\mathbf A\mathbf W, \mathbf E\mathbf W, \widetilde{\mathbf E}\mathbf W$).

We now consider that one of the weights (elements in the diagonal of $\mathbf W$), $w_i$, is a variable, so $\mathbf x$ depends on $w_i$ and we want to compute the derivatives of the entries in \(\mathbf x\) with respect to one of this weight \(w_i\). There is a:

\medskip
\noindent \textbf{Normalization issue.} \emph{When we modify the weights \(w_i\), the matrix \(\mathbf M=\mathbf B\mathbf W\) in Equation \eqref{eq.new} varies accordingly so that the Perron--Frobenius eigenvalue is always equal to $1$. In other words, we consider that the model is}
\begin{equation} \label{eq.normissue} \mathbf x=\mathbf M\mathbf x,\qquad \mathbf M=\frac{1}{\lambda}\mathbf B\mathbf W,\end{equation}

\noindent  \emph{where $\lambda$ is the Perron-Frobenius eigenvalue of $(\mathbf B\mathbf W)$ (and so depends on $w_i$), only one of the diagonal elements in the diagonal matrix $\mathbf W$ depends on $w_i$ and $\mathbf B$ is constant.}

\medskip

Before the following result, let us recall that, up to this moment, we have identified Perron-Frobenius eigenvector with \emph{right Perron-Frobenius eigenvectors}. There is an analogue theory for \emph{left Perron-Frobenius eigenvectors}. We omit details. We have that:

\begin{theorem}\label{theo.sensitivity}
Let \(\mathbf B\) be a symmetric, irreducible, and nonnegative constant matrix. Let \(\mathbf W\) be a diagonal matrix with diagonal entries \(w_1,\ldots,w_n\) and suppose that they are variables. Let us consider Eigenproblem \eqref{eq.normissue}. 

\begin{itemize}

\item[(a)] For $w_i\in (0,+\infty)$ and assuming that the rest of the weights are constant,  the Perron-Frobenius eigenvalue $\lambda$ of $(\mathbf B\mathbf W)$ is continuous with respect to $w_i$ and the partial derivative $\lambda'$ of $\lambda$ with respect to $w_i$ is well defined.

\item[(b)] For $w_i\in (0,+\infty)$, for every positive real number $N$, the vector $\mathbf x$ satisfying Equation \eqref{eq.normissue} and $\mathbf 1\mathbf x=N$ is continuous with respect to $w_i$ and the partial derivative $\mathbf x'$ of $\mathbf x$ with respect to $w_i$ is well defined.

\item[(c)] Then, it satisfies
\begin{equation} \label{eq.sensi}
(\mathbf I - \mathbf M)\, \mathbf x' = \mathbf M' \mathbf x,
\end{equation}
with
\begin{equation} \label{eq.Mprima}
\mathbf M' \;=\; \frac{1}{\lambda}\mathbf B\mathbf E_{ii}\;-\;\frac{\lambda'}{\lambda}\,\mathbf M,
\end{equation}
\begin{equation}\label{eq.lambdaprima}
\lambda' \;=\; x_i\,\mathbf u \mathbf B \mathbf e_i,
\end{equation}

where \(\mathbf E_{ii}\) is the matrix with a single $1$ in position \((i,i)\) and zeros elsewhere,  \(\mathbf x\) and \(\mathbf u\) are the right (column) and left (row) Perron-Frobenius eigenvectors of \(\mathbf B\mathbf W\), normalized so that \(\mathbf u \mathbf x=1\), and \(\mathbf e_i\) the \(i\)-th canonical basis vector.

\item[(d)] For every value $w_i\in (0,+\infty)$, there is a unique $\mathbf x$ satisfying Equation \eqref{eq.normissue} and a unique $\mathbf x'$ satisfying, simultaneously, the equations in Linear System \eqref{eq.sensi} and Equation \eqref{eq.ones}.
\end{itemize}
\end{theorem}

\begin{proof} The ingredients for a proof of Statements (a) and (b), as they appear here, can be found in Chapter 1 in \cite{R}. For the sake of completeness, let us provide the sketch of this proof. 

Let us consider the matrix $(\mathbf B\mathbf W)$. Let $p(x)$ be its characteristic polynomial. The coefficients of $p(x)$ are polynomials in $w_i$, so let us write, instead of $p(x)$, $p(x,w_i)$. For $w_i\in (0,+\infty)$, we know that $\lambda$ (that depends on $w_i$) is a simple root of $p(x)$, so $\frac{\partial p}{\partial x}(\lambda,w_i)\neq 0$. Then, by implicit differentiation, $\lambda'(w_i)$ exists and equals $\frac{\partial p}{\partial w_i}(\lambda,w_i)/(\frac{\partial p}{\partial x}(\lambda,w_i))$. This proves Statement (a).

For every $w_i\in(0,+\infty)$, the kernel of $(\mathbf B\mathbf W-\lambda\mathbf I)$ has rank 1. The vector $\widetilde{\mathbf x}=[\Gamma_{n1},\ldots,\Gamma_{nn}]^T$, where $\Gamma_{nj}$ denotes the corresponding cofactor in $(\mathbf B\mathbf W-\lambda\mathbf I)$, is  in the kernel of $(\mathbf B\mathbf W-\lambda\mathbf I)$, its entries are non-simultaneously zero and they are differentiable with respect to  $w_i$ (details appear in Chapter 1 in \cite{R} and in the \emph{Proof of Rellich's Lemma} in \cite{bellido2022rellich}). So the vector $\mathbf x=\widetilde{\mathbf x}/(\mathbf 1\widetilde{\mathbf x})$ is also differentiable with respect to $w_i$. This proves Statement (b).

Differentiating the equation \(\mathbf x = \mathbf M\mathbf x\) with respect to \(w_i\) and rearranging terms we get Equation \eqref{eq.sensi}. Taking derivatives in \(\mathbf M = \frac{1}{\lambda}\mathbf B\mathbf W\) with respect to \(w_i\) (recall that \(\mathbf B\) is constant) we obtain Equation \eqref{eq.Mprima}.

Let \(\mathbf u\) be the left Perron-Frobenius eigenvectors of \(\mathbf B\mathbf W\), normalized so that \(\mathbf u \mathbf x=1\). A proof similar to the one in Statement (b) shows that $\mathbf u'$ is also well defined.

We have that
\[
\lambda = \mathbf u(\mathbf B\mathbf W)\mathbf x.
\]
\noindent Differentiating the previous equation yields
\[
\lambda'= \mathbf u'(\mathbf B\mathbf W)\mathbf x \;+\;\mathbf u(\mathbf B\mathbf E_{ii})\mathbf x\;+\;\mathbf u(\mathbf B\mathbf W)\mathbf x'.
\]
\noindent Differentiating \(\mathbf u\mathbf x=1\), we obtain $\mathbf u'\mathbf x+\mathbf u\mathbf x'=0$. Substituting into the previous expression gives Equation \eqref{eq.lambdaprima}. This proves Statement (c).

For Statement (d), note that, for each value of $w_i\in U$, $(\mathbf I-\mathbf M)$ has rank $n-1$ (since the Perron–Frobenius eigenvalue $1$ is algebraically simple and the Linear System in Equation \eqref{eq.sensi} is compatible, as a consequence of Statement (a)). We have that $\ker(\mathbf I-\mathbf M)=\text{span}\{\mathbf x\}$ and $\mathbf 1 \mathbf x\neq 0$. Thus, if we add the row $\mathbf 1$ to $(\mathbf I-\mathbf M)$ we obtain a matrix of full rank $n$. Therefore, the augmented system has a unique solution. This proves Statement (d).

\end{proof}

Let us remark that, if we want to compute partial derivatives with respect to more than one weight $w_1,\ldots, w_n$, all the linear systems involved share the same coefficient matrix. This is computationally relevant.

\newpage
\subsection*{Sensitivity analysis with respect to the weights in an eigenvector model for our toy example}

For the example introduced in Section \ref{section.eigen}, let us recall that the adjacency matrix and the matrix of weights (as it was simulated to be estimated in Section \ref{section.flow}) are
\begin{equation}\label{eq.Mtoy}\mathbf A=\begin{bmatrix}0 & 1 & 0 \\ 1 & 0 & 1 \\ 0 & 1 & 0 \end{bmatrix},
\qquad\mathbf W=\begin{bmatrix}w_1& 0 & 0\\ 0 & w_2 & 0 \\ 0 &0 & w_3 \end{bmatrix}=\begin{bmatrix}2 & 0 & 0 \\ 0& \frac{1}{3}& 0  \\ 0& 0& 1\end{bmatrix}.   \end{equation}

Let us perform the sensitivity analysis of the model $\mathbf x=\mathbf A\mathbf W\mathbf x$ with respect to these parameters $w_1,w_2,w_3$, assuming that $\mathbf 1\mathbf x=500$ (and it is constant).

We need to compute the derivatives of $\mathbf x$ with respect to each of the variables $w_1,w_2,w_3$. To do so, we need to solve three linear systems as the one in Equation \eqref{eq.sensi}, but adding the final Equation \eqref{eq.ones} at each of them. After this, we obtain three linear systems, all of them with the same $4\times 3$ coefficient matrix:
$$\left[\begin{array}{c}\mathbf I-\mathbf M\\ \hline \mathbf 1\end{array}\right]=\left[\begin{array}{c c c}1 & -1/3 & 0 \\ -2 & 1 & -1 \\ 0 & -1/3 & 1 \\ \hline  1 & 1 & 1 \end{array}\right]. $$

\noindent If we arrange the solutions in a matrix whose entry $(i,j)$ is the derivative of $x_i$ with respect to $w_j$, we get:
$$\begin{bmatrix}
-10 & 90 & -10 \\[1mm]
20 & -180 & 20 \\[1mm]
-10 & 90 & -10
\end{bmatrix}.
$$

\noindent Finally, if from the information in the previous matrix we compute the elasticities with Equation \eqref{eq.defielas} and then we arrange these elasticities in another matrix following the same pattern, we get the result in Figure \ref{fig.heatmap}. We also present the information in a bar chart in Figure \ref{fig.elasticities}.

\begin{figure}[h]
\centering
\includegraphics[width=0.6\textwidth]{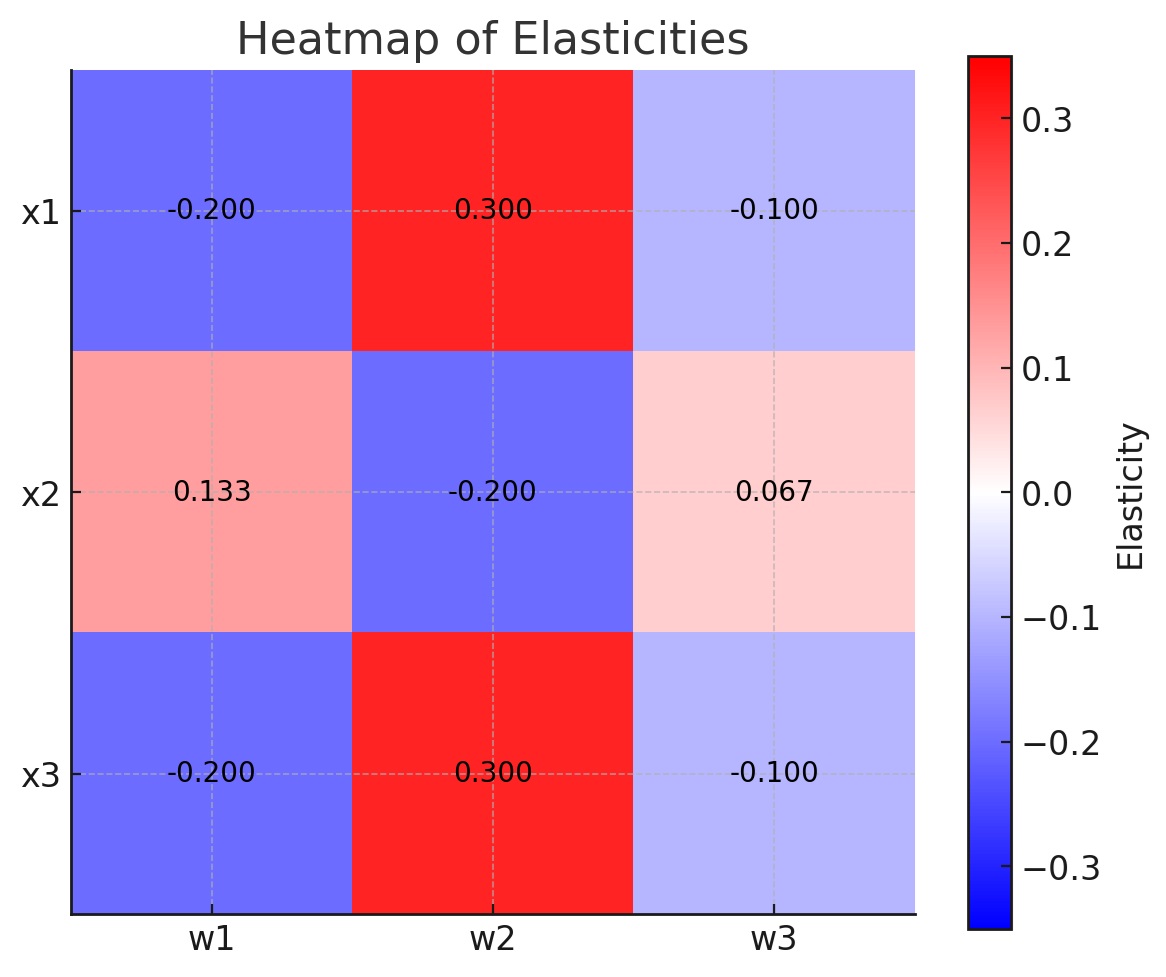}
\caption{Heatmap corresponding to the elasticities.}\label{fig.heatmap}
\end{figure}

\begin{figure}[h]
\centering
\includegraphics[width=0.6\textwidth]{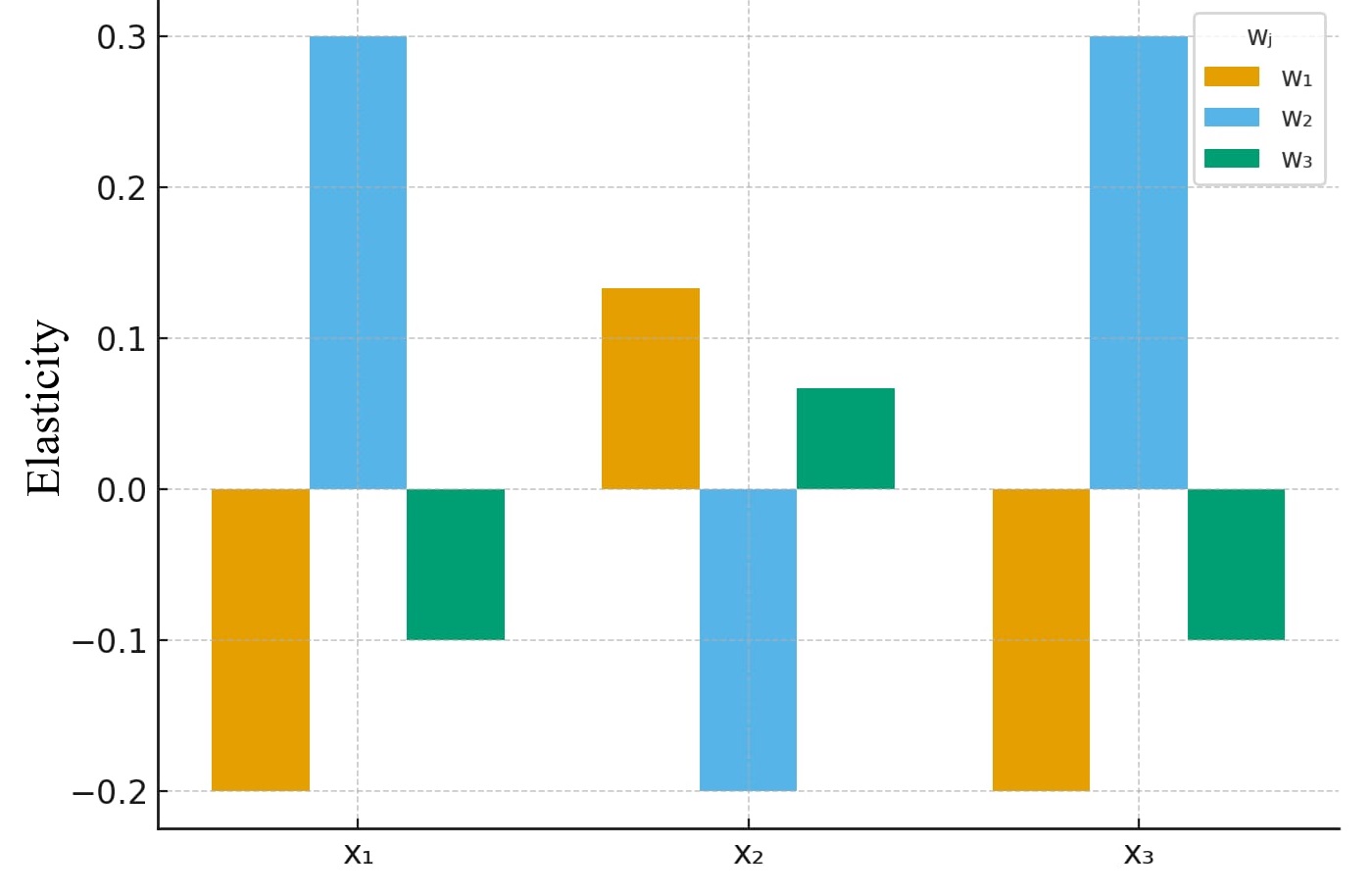}
\caption{Bar chart for the elasticities of $x_1,x_2,x_3$ with respect to $w_1,w_2,w_3$, respectively.}\label{fig.elasticities}
\end{figure}

As we can see:

\begin{itemize}

\item No elasticity is very large (all smaller than 0.5 in absolute value), indicating \emph{moderate effects}.

\item Mixed signs show that factors can act as \emph{incentives for some variables and disincentives for others}.

\item  $x_1,x_3$ depend positively on $w_2$ and negatively on $w_1,w_3$. 

\item \(x_2\) increases with \(w_1\) and \(w_3\) (though modestly) and decreases with \(w_2\) (that has a stronger negative effect).  

\item \(w_2\) is an important factor: it has opposite effects on \(x_1,x_3\) (positive) and \(x_2\) (negative).

\end{itemize}

\newpage
\subsection*{Sensitivity analysis for shifted eigenvector models}

Let us consider in this subsection eigenproblems of the type

\begin{equation} \label{eq.finalboss} \mathbf x=\mathbf M\mathbf x+\mathbf f,\qquad \mathbf M=\mathbf B\mathbf W,\end{equation}

\noindent for $\mathbf B$ being a symmetric, irreducible and nonnegative matrix, $\mathbf W$ being the diagonal matrix containing the (positivie) weights (if we do not want to use weights, we take $\mathbf W=\mathbf I$), $\rho(\mathbf M)<1$ and nonnegative $\mathbf f$.

This time, we may consider as a variable one of the weights or one of the entries in $\mathbf f$. To simplify the writing, let us name this variable as $t$. We have the following:

\medskip
\noindent \textbf{Normalization issue.} \emph{When we modify $t$, the matrix $\mathbf M=\mathbf B\mathbf W$ in Equation \eqref{eq.finalboss} varies accordingly so that}
$$\mathbf 1\mathbf x=\mathbf 1\mathbf M\mathbf x+\mathbf 1\mathbf f$$

\noindent  \emph{remains constantly equal to some positive value $N$. In other words, we are actually considering the model}
\begin{equation}\label{eq.normissue2}  \mathbf x=\mathbf M\mathbf x+\mathbf f,\qquad \mathbf M=\mu \mathbf B\mathbf W \end{equation}

\noindent \emph{where} 
\begin{equation}\label{eq.mu}\mu=\frac{N-\mathbf 1\mathbf f}{\mathbf 1\mathbf B\mathbf W\mathbf x} \end{equation}

\noindent \emph{(and so depends on $t$), only one of the entries among $\mathbf W,\mathbf d$ depend on $t$, and $\mathbf B$ is constant.}

\medskip

We have the following:

\begin{theorem} \label{theo.2} Let $\mathbf B$ be a symmetric, irreducible and nonnegative constant matrix. Let $\mathbf W$ be a diagonal matrix with diagonal entries $w_1,\ldots, w_n$. Let $\mathbf f=[f_1,\ldots, f_n]^T$ be a nonnegative and nontrivia vector. Suppose that $\rho(\mathbf B\mathbf W)<1$. Suppose that one of the parameters $t$ among $w_1,\ldots, w_n,f_1,\ldots, f_n$ is variable and the rest of them are constant.  Let $\mathbf W',\mathbf f',\mu'$ denote the partial derivatives of $\mathbf W,\mathbf f,\mu$, respectively, with respect to $t$.

Let $t_0$ be some value for the variable $t$. We denote by $\mathbf W_0,\mathbf M_0,\mathbf f_0$ to the matrices and vector obtained from $\mathbf W,\mathbf M,\mathbf f$ evaluating $t=t_0$ and let $\mathbf x_0$ be the solution of $\mathbf x_0=\mathbf B\mathbf W_0\mathbf x_0+\mathbf f_0$, .

Let us consider the Problem \eqref{eq.normissue2}, where $\mu$ is given as in Equation \eqref{eq.mu} and $N=\mathbf 1\mathbf x_0$.

\noindent (a)  There exists some neighborhood $U$ of $t_0$ such that, for every $t\in U$, the vector $\mathbf x$ is continuous with respect to $t$ and the partial derivative $\mathbf x'$ of $\mathbf x$ with respect to $t$ is well defined.

\noindent (b) For $t\in U$,
\begin{equation}\label{eq.magdalena} \mathbf C\mathbf x'=\mathbf D\mathbf x, \end{equation}

$$\text{for} \quad \begin{cases}\mathbf C=[(\mathbf x-\mathbf f)\mathbf 1\mathbf B\mathbf W+(\mathbf 1\mathbf B\mathbf W\mathbf x)\mathbf I-(N-\mathbf 1\mathbf f)\mathbf B\mathbf W],\\
\mathbf D=[-\mathbf 1\mathbf f'\mathbf B\mathbf W+(N-\mathbf 1\mathbf f)\mathbf B\mathbf W'+\mathbf f\mathbf 1\mathbf B\mathbf W'+\mathbf f'\mathbf 1\mathbf B\mathbf W-(\mathbf 1\mathbf B\mathbf W'\mathbf x)\mathbf I].\end{cases}$$

\noindent Recall that one of $\mathbf W',\mathbf f'$ is zero.

\noindent (c) Evaluating $t=t_0$, there is a unique solution $\mathbf x'$ for the Linear System \eqref{eq.magdalena} in the indeterminate $\mathbf x'$.

\end{theorem}

\begin{proof} Let $U$ be an open interval containing $t_0$, such that if $t\in U$, then $\mathbf M$ is irreducible, nonnegative, and has spectral radius less than 1 and  $\mathbf f$ is nonnegative and nontrivial. Statement (a) is a consequence of Theorem \ref{theo.leontief} and of Cramer's rule, applied to the consistent and independent linear system $(\mathbf I-\mathbf M)\mathbf x=\mathbf f$.

From Equations \eqref{eq.normissue2} and \eqref{eq.mu} we get
$$(\mathbf 1\mathbf B\mathbf W\mathbf x)\mathbf x=(N-\mathbf 1\mathbf f)\mathbf B\mathbf W\mathbf x+(\mathbf 1\mathbf B\mathbf W\mathbf x)\mathbf f. $$

\noindent Taking derivatives with respect to $t$ and re-arranging the terms, we obtain Equation \eqref{eq.magdalena} and prove Statement (b).






For $t=t_0$, we have that 
$$N-\mathbf 1\mathbf f_0=\mathbf 1\mathbf M_0\mathbf x_0,\qquad \mathbf x_0-\mathbf f_0=\mathbf M_0\mathbf x_0.$$

\noindent  So if we evaluate $\mathbf C$ at $t=0$, we obtain
$$\mathbf C_0= [(\mathbf M_0\mathbf x_0)\mathbf 1\mathbf M_0+(\mathbf 1\mathbf M_0\mathbf x_0)\mathbf I-(\mathbf 1\mathbf M_0\mathbf x_0)\mathbf M_0].$$

Let us introduce the short-hand notation
\[
\mathbf v := \mathbf M_0\mathbf x_0,\qquad
s := \mathbf{1}\mathbf M_0 \mathbf x_0 ,\qquad
\mathbf w := \mathbf{1}\mathbf M_0,
\]

\noindent so
$$\mathbf C_0 = \mathbf v\mathbf w + s(\mathbf I-\mathbf M_0).$$

Since the spectral radius of $\mathbf M_0$ is less than 1, $(\mathbf I-\mathbf M)$ is invertible. So
$$\mathbf C_0 = s(\mathbf I-\mathbf M)(\mathbf I+\frac{1}{s}(\mathbf I-\mathbf M)^{-1}\mathbf v\mathbf w).$$

The factor \(s(\mathbf I-\mathbf M)\) is invertible (product of the nonzero scalar \(s\) and the invertible matrix \((\mathbf I-\mathbf M)\)), so \(\mathbf C_0\) is invertible if and only if the determinant $\Delta$ of the factor in parentheses is invertible. The latter is a rank-one perturbation of the identity. Let $\mathbf u := s^{-1}(\mathbf I-\mathbf M)^{-1} \mathbf v$. Then the matrix in the parentheses equals $\mathbf I+\mathbf u\mathbf w$. Apply the \emph{Sherman--Morrison Formula} for determinant:
\[
\Delta=\; 1 + \mathbf w \mathbf u .
\]

\noindent Substituting back \(\mathbf w=\mathbf{1}\mathbf M\) and \(\mathbf v=\mathbf M\mathbf x\) gives
\[
 \Delta=\Big(1 + s^{-1}\,\mathbf{1}(\mathbf I-\mathbf M)^{-1}\mathbf M \mathbf x\Big).
\]

\noindent A product of nonnegative matrices is always nonnegative, so $\Delta>0$.







\end{proof}

\subsection*{Sensitivity analysis with respect to forced occupancy coefficents in a shifted eigenvector model for our toy example}

Consider the eigenvector model with offset:
\begin{equation} \label{eq.matrix}
\mathbf x = \alpha
\begin{bmatrix} 
0 & \frac{1}{3} & 0 \\ 
2 & 0 & 1 \\ 
0 & \frac{1}{3} & 0
\end{bmatrix} \mathbf x + 
\begin{bmatrix} 0 \\ 100 \\ 100 \end{bmatrix}.
\end{equation}

This model corresponds to the toy example introduced earlier. 

\begin{itemize}

\item The matrix on the right-hand side of Equation \eqref{eq.matrix} is proportional to $\mathbf M = \mathbf A \mathbf W$, for the matrices $\mathbf A,\mathbf W$ appearing in Equation \eqref{eq.Mtoy}. The parameters in $\mathbf W$ were estimated on a day without classes, and we assume that the \emph{relative attractiveness} of the spaces does not change in a day with classes.

\item We consider a forced occupancy equal to 100 at nodes $V_2$ and $V_3$ (corresponding to scheduled classes). 

\item We assume that total number of students remains constant regardless of classes. So the parameter $\alpha$ is chosen so that the sum of the entries in $\mathbf x$ equals 500. In this case, $\alpha = \frac{3}{5}$, and the vector of centralities determined by Equation \eqref{eq.matrix} is
\[
\mathbf x = 
\begin{bmatrix} 50 \\ 250 \\ 150 \end{bmatrix}.
\]
\end{itemize}

Next, consider 
\begin{equation}\label{eq.final}
\mathbf x = \mu\alpha
\begin{bmatrix} 
0 & \frac{1}{3} & 0 \\ 
2 & 0 & 1 \\ 
0 & \frac{1}{3} & 0
\end{bmatrix} \mathbf x +
\begin{bmatrix} 0 \\ 100 \\ f_3 \end{bmatrix},
\end{equation}
where $f_3$ is variable (is the variable $t$ in the previous theorem) and $\mu$ depends on $f_3$ to maintain $\mathbf 1 \mathbf x = 500$.


Performing the sensitivity analysis using the previous theorem, we obtain the derivative of $\mathbf x$ with respect to $f_3$ which, evaluated at $f_3 = 100$, is
\[
\mathbf x'(100) \approx 
\begin{bmatrix} -0.278 \\ -0.443 \\ 0.772 \end{bmatrix}.
\]

This is consistent with expectations: increasing the forced occupancy at $V_3$ raises the centrality at $V_3$, while the model predicts a larger negative impact on $V_2$ than on $V_1$.

\section{Final remarks} \label{section.inveigen}


As explained in Section \ref{section.eigen}, one possibility is to define \emph{centrality} for all vertices, weighting contributions by distance so that more distant nodes contribute less \cite{Nourian2016}. But the models in that section are not the only possibility, as mentioned in \cite{Newman2010}. Another option is related to a concept known as \emph{integration} in the literature of urban networks.






\medskip

\noindent\textbf{Shifted inverse eigenvector models for centrality.} \emph{Let $\mathbf M$ be a nonnegative, symmetric, fully indecomposable matrix suitably associated to an urban network $(V,E)$. We define the \textbf{inverse eigenvector of centralities associated to $\mathbf M$} as the (unique up to scaling) positive inverse eigenvector of $\mathbf M$, that is, the (unique up to scaling) positive vector $\mathbf x$ satisfying}
\begin{equation}\label{eq.inv}
\lambda \overline{\mathbf x}=\mathbf M \, {\mathbf x}+\mathbf f,
\end{equation}

\noindent \emph{for some positive real $\lambda$,  where $\overline{\mathbf x}$ denotes the column vector containing the reciprocals of the entries in $\mathbf x$.}

\medskip

For example, we can take $\mathbf M=\mathbf D$, where $\mathbf D$ denotes the \textbf{distance matrix} of an urban network $(V,E)$, that is, $\mathbf D$ is the $n \times n$ matrix with entries $d_{ij}$ in the notation of Section \ref{section.eigen}.


Such problems are particular cases of \emph{nonlinear eigenproblems}. For the case $\mathbf f=\mathbf 0$ are sometimes referred to in the literature as \emph{inverse eigenvector problems} or \emph{power equations}. While not widely known among non-specialist, they arise naturally in certain applied mathematics contexts, such as the \emph{Theory of Power in Networks} or the \emph{Optimal Plank Theorem} (see, for example, \cite{BF,O}). This has motivated their study, although results are scattered across several articles and a unified treatment has yet to emerge.

The existence and uniqueness of a positive solution of Problem \eqref{eq.inv} is studied by the general theory of nonlinear Perron--Frobenius operators \cite{nussbaum1986,lemmens2012}. 
This theory extends classical results to certain nonlinear maps, including the reciprocal eigenvector problem. The reader may also find some information about the connection of Problem \eqref{eq.inv} to the \emph{Balancing Problem}, and how to solve it numerically in \cite{BF}. Let us include the most important facts about this.

\begin{theorem}[Combination of Theorems 1, 2 in \cite{BF}] Let $\mathbf M$ be a symmetric nonnegative and fully indecomposable square matrix.  The solution of the Problem \eqref{eq.inv} exists and it is unique, for each fixed value $\lambda$.

\end{theorem}

To solve numerically the problem above, we refer the reader to \cite{BF,SK}. The iterative scheme described there is essentially the \emph{Sinkhorn--Knopp Algorithm} for matrix balancing, which guarantees convergence under full indecomposability.

\medskip

\noindent \textbf{Remark.} \emph{Note that not every matrix in the context of inverse eigenvector models for centrality satisfy the hypothesis for the Theorem. Moreover, we can also consider inverse eingenvector problems with shifting. This case can be reduced to an inverse eigenvector problems where the corresponding matrix does not satisfy the hypothesis in the theorem. We omit details.}


\backmatter

\medskip
\noindent \textbf{Author contributions} All authors contributed equally to this work.

\medskip

\noindent \textbf{Funding} The current investigation received no funding from any sources.

\section*{Statements and declarations}

The authors declare that they have no competing interests that could have influenced the work reported in this article.

\bibliography{sn-bibliography}

@article{agryzkov2019centrality,
  title={A centrality measure for urban networks based on the eigenvector centrality concept},
  author={Agryzkov, Ivan and Tortosa, Leandro and Vicent, Jos{\'e} F. and Wilson, Richard},
  journal={Environment and Planning B: Urban Analytics and City Science},
  volume={46},
  number={4},
  pages={668--689},
  year={2019},
  publisher={SAGE Publications, London, UK},
  doi={10.1177/2399808317724444},
  url={https://doi.org/10.11772399808317724444}
}

@article{BP,
  title={Existence and uniqueness of solutions for Leontief's Input--Output Model, graph theory and sensitivity analysis},
  author={Bellido, Jos{\'e} Carlos and Prieto-Mart{\'\i}nez, Luis Felipe},
  journal={Linear and Multilinear Algebra},
  volume={73},
  number={9},
  pages={2103--2124},
  year={2025},
  publisher={Taylor \& Francis, London, UK},
  doi={10.1080/03081087.2024.2313638},
  url={https://doi.org/10.1080/03081087.2024.2313638}
}

@article{bellido2022rellich,
  title={A Rellich's result revisited and sensitivity of solutions of parametrized linear systems},
  author={Bellido, Jos{\'e} Carlos and Prieto-Mart{\'\i}nez, Luis Felipe},
  journal={arXiv preprint arXiv:2301.13164},
  year={2022}
}

@book{bjorck1996numerical,
  title={Numerical Methods for Least Squares Problems},
  author={Bj{\"o}rck, {\AA}ke},
  year={1996},
  publisher={SIAM},
  address={Philadelphia, PA, USA}
}

@article{bonacich1972factoring,
  title={Factoring and weighting approaches to status scores and clique identification},
  author={Bonacich, Phillip},
  journal={Journal of Mathematical Sociology},
  volume={2},
  number={1},
  pages={113--120},
  year={1972},
  publisher={Taylor \& Francis},
  doi={10.1080/0022250X.1972.9989806}
}

@article{bonacich1987power,
  title={Power and centrality: A family of measures},
  author={Bonacich, Phillip},
  journal={American Journal of Sociology},
  volume={92},
  number={5},
  pages={1170--1182},
  year={1987},
  publisher={University of Chicago Press},
  doi={10.1086/228631}
}

@article{BF,
  title={A theory on power in networks},
  author={Bozzo, Enrico and Franceschet, Massimo},
  journal={Communications of the ACM},
  volume={59},
  number={11},
  pages={75--83},
  year={2016},
  publisher={ACM, New York, NY, USA},
  doi={10.1145/2934665},
  url={https://doi.org/10.1145/2934665}
}

@book{Chung1997,
  title={Spectral Graph Theory},
  author={Chung, Fan R. K.},
  year={1997},
  publisher={American Mathematical Society},
  address={Providence, RI, USA}
}

@article{Crucitti2006,
  title={Centrality measures in spatial networks of urban streets},
  author={Crucitti, Paolo and Latora, Vito and Porta, Sergio},
  journal={Physical Review E},
  volume={73},
  number={3},
  pages={036125},
  year={2006},
  publisher={APS},
  doi={10.1103/PhysRevE.73.036125},
  url={https://doi.org/10.1103/PhysRevE.73.036125}

}

@book{erlander1990gravity,
  title={The gravity model in transportation analysis: theory and extensions},
  author={Erlander, Sven and Stewart, Neil F.},
  year={1990},
  publisher={VSP},
address={Utrecht, Netherlands}
}

@book{ersel2011linear,
  title={Linear algebra for economists},
  author={Ersel, Hasan and Piontkovski, Dmitri},
  year={2011},
  publisher={Springer Science \& Business Media},
  address={Berlin, Germany}
}

@book{Golub2013,
  title={Matrix Computations},
  author={Golub, Gene H. and Van Loan, Charles F.},
  edition={4th},
  year={2013},
  publisher={Johns Hopkins University Press},
  address={Baltimore, MD, USA},
  isbn={9781421407944}
}

@book{HH,
  author={Hillier, Bill and Hanson, Julienne},
  title={The Social Logic of Space},
  publisher={Cambridge University Press},
  year={1989},
  address={Cambridge, UK}
}

@book{horn2012matrix,
  title={Matrix Analysis},
  author={Horn, Roger A. and Johnson, Charles R.},
  year={2012},
  edition={2nd},
  publisher={Cambridge University Press},
  address={Cambridge, UK}
}

@article{JC,
  title={Topological analysis of urban street networks},
  author={Jiang, Bin and Claramunt, Christophe},
  journal={Environment and Planning B: Planning and design},
  volume={31},
  number={1},
  pages={151--162},
  year={2004},
  publisher={SAGE Publications Sage UK: London, England},
doi={10.1068/b306}
}

@article{Katz1953,
  title={A new status index derived from sociometric analysis},
  author={Katz, Leo},
  journal={Psychometrika},
  volume={18},
  number={1},
  pages={39--43},
  year={1953},
  publisher={Springer},
  doi={10.1007/BF02289026}
}

@book{lemmens2012,
  author={Lemmens, Bas and Nussbaum, Roger D.},
  title={Nonlinear Perron--Frobenius Theory},
  publisher={Cambridge University Press},
  year={2012},
  address={Cambridge, UK}
}

@book{Newman2010,
  title={Networks: An introduction},
  author={Newman, Mark},
  year={2010},
  publisher={Oxford University Press},
  address={Oxford, UK},
  isbn={9780199206650}
}

@article{niu2023technical,
  title={The technical approach of using mobile positioning data to support urban population size monitoring},
  author={Niu, Xinyi and Lin, Shijia and Qin, Sixian and Yue, Yufeng},
  journal={Frontiers of Urban and Rural Planning},
  volume={1},
  number={1},
  year={2023},
  publisher={Springer},
doi={10.1007/s44243-023-00013-y}
}

@article{Nourian2016,
  title={Spectral modelling for spatial network analysis},
  author={Nourian, Pirouz and Rezvani, Samaneh and Sariyildiz, Sevil and van der Hoeven, Franklin},
  journal={Proceedings of the Sumposium on Simulation for Architecture and Urban Design (simAUD 2016)},
  volume={11},
  pages={103--110},
  year={2016}
}

@article{nussbaum1986,
  author={Nussbaum, Roger D.},
  title={Convexity and log convexity for the spectral radius},
  journal={Linear Algebra and its Applications},
  volume={73},
  pages={59--122},
  year={1986},
  doi={10.1016/0024-3795(86)90233-8}
}

@article{O,
  title={An optimal plank theorem},
  author={Ortega-Moreno, {\'O}scar},
  journal={Proceedings of the American Mathematical Society},
  volume={149},
  number={3},
  pages={1225--1237},
  year={2021},
  doi={10.1090/proc/15228},
  url={https://doi.org/10.1090/proc/15228}
}

@article{Porta2006,
  title={The network analysis of urban streets: A dual approach},
  author={Porta, Sergio and Crucitti, Paolo and Latora, Vito},
  journal={Physica A: Statistical Mechanics and its Applications},
  volume={369},
  number={2},
  pages={853--866},
  year={2006},
  publisher={Elsevier},
  doi={10.1016/j.physa.2005.12.063},
  url={https://doi.org/10.1016/j.physa.2005.12.063}

}

@book{R,
  title={Perturbation theory of eigenvalue problems},
  author={Rellich, Franz},
  year={1969},
  publisher={CRC Press},
address={New York, NY, USA}
}

@book{searle1997linear,
  title={Linear Models},
  author={Searle, Shayle R.},
  year={1997},
  publisher={Wiley},
  address={Hoboken, NJ, USA}
}

@book{Seneta2006,
  title={Non-negative Matrices and Markov Chains},
  author={Seneta, Eugene},
  year={2006},
  publisher={Springer},
  address={New York, NY, USA}
}

@article{SK,
  title={Concerning nonnegative matrices and doubly stochastic matrices},
  author={Sinkhorn, Richard and Knopp, Paul},
  journal={Pacific Journal of Mathematics},
  volume={21},
  number={2},
  pages={343--348},
  year={1967},
  doi={10.2140/pjm.1967.21.343},
  url={https://doi.org/10.2140/pjm.1967.21.343}
}

@book{strang2009introduction,
  title={Introduction to Linear Algebra},
  author={Strang, Gilbert},
  year={2022},
  edition={6th},
  publisher={Wellesley-Cambridge Press},
  address={Wellesley, MA, USA}
}

@book{strang2016la4e,
  title={Linear Algebra and Learning from Data},
  author={Strang, Gilbert},
  year={2019},
  publisher={Wellesley-Cambridge Press},
  address={Wellesley, MA, USA}
}

@book{varga2009matrix,
  title={Matrix Iterative Analysis},
  author={Varga, Richard S.},
  year={1962},
  publisher={Springer},
  address={Berlin, Germany}
}

\end{document}